\documentclass[onecolumn,draftcls]{IEEEtran}
\usepackage[utf8]{inputenc}
\usepackage{textcomp}

\usepackage{cite,color}
\usepackage{graphicx}
\usepackage{amsmath,mathtools}
\usepackage{longtable,tabularx}
\usepackage{subcaption}

\usepackage{amsthm}
\newtheorem{theorem}{Theorem}
\newtheorem{assumption}{Assumption}
\newtheorem{lemma}{Lemma}

\title{\bf NMPC-Based Cooperative Strategy For A Target Pair To Lure Two Attackers Into Collision}

\author{Amith Manoharan and P.B. Sujit %
	\thanks{Amith Manoharan is a Graduate Student at IIIT Delhi, New Delhi -- 110020, India. email: amithm@iiitd.ac.in}
	\thanks{P.B. Sujit is Associate Professor at IISER Bhopal, Bhopal -- 462066, India. email: sujit@iiserb.ac.in}}

\begin{document}
	
	\maketitle
	\IEEEpeerreviewmaketitle
	
	\begin{abstract}
		
		This paper presents a cooperative target defense strategy using nonlinear model-predictive control (NMPC) framework for a two--targets two--attackers (2T2A) game. The 2T2A game consists of two attackers and two targets. Each attacker needs to capture a designated target individually. However, the two targets cooperate to lure the attackers into a collision. We assume that the cooperative target pair do not have perfect knowledge of the attacker states, and hence they estimate the attacker states using an extended Kalman filter (EKF). The NMPC scheme computes closed-loop optimal control commands for the targets while respecting imposed state and control constraints. Theoretical analysis is carried out to determine regions that will lead to the targets' survival, given the initial positions of the attacker and target agents. Numerical simulations are carried out to evaluate the performance of the proposed NMPC-based strategy for different scenarios.
	\end{abstract}
	
	
	
	
	
	

	\section{Introduction}\label{sec:intro}
Differential games find an important role in aerial combat scenarios involving missiles and aircraft and in several other security games. Isaacs~\cite{isaacsbook} pioneered the work on modeling these scenarios as pursuit-evasion games with the missile as pursuer and the aircraft as the evader. The simplest of these games is a two-agent game with a single pursuer and evader \cite{isler2005randomized,bopardikar2008discrete,jagat2017nonlinear,sunkara2018pursuit}. The objective of the pursuer is to capture the evader while the evader tries to escape from the pursuer. There are several variants of two-agent pursuit-evasion games \cite{merz1974homicidal,alpern2008princess,karnad2009lion}. The two-agent game has also been extended to have multiple pursuers against a single evader in \cite{c2,sun2017multiple,ramana2017pursuit,pachter2019two,pachter2020cooperative}.

Cooperation among the pursuers or evaders allows one to determine a variety of solutions for the multi-agent pursuit-evasion games. One such game is the three-agent target-attacker-defender (TAD) game \cite{perelman2011cooperative}. In this game, the attacker tries to capture the target, while the defender tries to intercept the attacker before it can reach the target. The target and the defender cooperate to increase their payoff. The TAD problem has been solved using different techniques like linear-quadratic formulations \cite{perelman2011cooperative,c8,garcia2020defense}, command-to-line-of-sight (CLOS) \cite{c9,ratnoo2012guidance,yamasaki2013modified}, sliding-mode control \cite{kumar2017cooperative}, optimal control \cite{rubinsky2014three,c11}, differential game theory \cite{c12,garcia2017cooperative,garcia2018design,garcia2021complete}, and model-predictive control \cite{manoharan2019nmpc}. Garcia~et~al.~[28] present a three-agent game in 3D different from the generic TAD game where a team of two pursuers guards a target against an evader with the same speed as the pursuers.

A natural extension of the three-agent TAD game is the four-agent game. 
Casbeer et al. \cite{casbeer2018target} introduce a four-agent game by introducing an additional defender to the TAD game proposed in \cite{garcia2017cooperative}. 
Garcia et al. \cite{garcia2020two} formulate a linear-quadratic differential game that can be implemented in a state feedback form for two attackers that are launched against a stationary target while two interceptors are fired to defend the target. A variant of \cite{garcia2020two} is the assignment problem for a multi-pursuer multi-evader differential game \cite{garcia2020multiple}, where optimal assignments of pursuers to evaders are studied. Another interesting formulation on two pursuers and two targets is given in \cite{tan2018cooperative}, where the targets use a state-dependent-Riccati-equation (SDRE) based approach to lure the pursuers into collision for the targets' survival. The initial positions of the pursuers and the targets have an impact on the ability of the SDRE formulation to lure the pursuers into collision. Several simulations with different initial conditions are conducted to show this impact. However, no theoretical analysis is presented. A two-on-two scenario solved using linear-quadratic differential game strategy was presented in \cite{liang2020guidance} for intercepting a spacecraft with active defense.  

Inspired by \cite{tan2018cooperative}, we advance the two--targets two--attacker (2T2A) game formulation to use nonlinear model-predictive control (NMPC) for the computation of optimal control commands. We show that the 2T2A can be formulated as a game of kind~\cite{isaacsbook} in which the outcome is determined by the initial positions of the attackers and the targets. To facilitate this outcome, we determine the escape region for the targets theoretically.

The strategies presented in \cite{garcia2020defense,casbeer2018target,garcia2020two,liang2020guidance} assume perfect information about the attacker states and guidance laws employed by it. This paper relaxes the above assumption using an extended Kalman filter (EKF) for the attacker state estimation. Also, since the NMPC is closed-loop and the control inputs are determined with the current state estimates while also taking into account the future trajectories, it provides the flexibility to adapt to uncertain environments and unknown attacker guidance laws, which is not possible in the case of open-loop optimal control strategies like \cite{casbeer2018target} where the solution is predetermined. These factors indicate that the NMPC has an additional advantage of real-world implementability. The formulations which use closed-loop solutions \cite{garcia2020defense,garcia2020two,liang2020guidance,tan2018cooperative} linearize the system resulting in approximation errors. Our approach is based on the nonlinear model of the system and hence is superior. Also, the inherent constraint handling of NMPC helps to design controls that strictly adhere to the specified bounds. Tan et al.~\cite{tan2018cooperative} use simulations to show the region of escape while we present the theoretical framework that determines whether the attackers can be lured to collide or not. 

The main contributions of this article are 1) NMPC formulation for the 2T2A game without any assumptions on the attacker states or guidance laws, (2) theoretical analysis of the escape region for the four-agent game, and (3) validation of the proposed approach and the escape region results through numerical simulations under different conditions.

The rest of the paper is organized as follows. The problem formulation is given in Section~\ref{sec:problem}. The nonlinear model-predictive control scheme is explained in Section~\ref{sec:nmpc}. Analysis of the target escape region is given in Section~\ref{sec:escape}. Simulation results are presented in Section~\ref{sec:results}, and the conclusions are given in Section~\ref{sec:conclusions}. 
	\section{Problem formulation}\label{sec:problem}
Consider a pursuit-evasion game between a pair of attackers $(A1,A2)$ pursuing a pair of targets $(T1,T2)$ as shown in Fig.~\ref{fig:E_geometry}. We call this the two--targets two--attackers (2T2A) game and is inspired from~\cite{tan2018cooperative}, where the objective of the target pair is to maneuver in such a way that the attackers collide with each other, ensuring the survival of the targets. The target pair is cooperative, whereas the attackers act individually. We consider the following assumptions in the formulation of this problem.

\begin{figure}
	\centering
	\includegraphics[width=9cm,height=8cm]{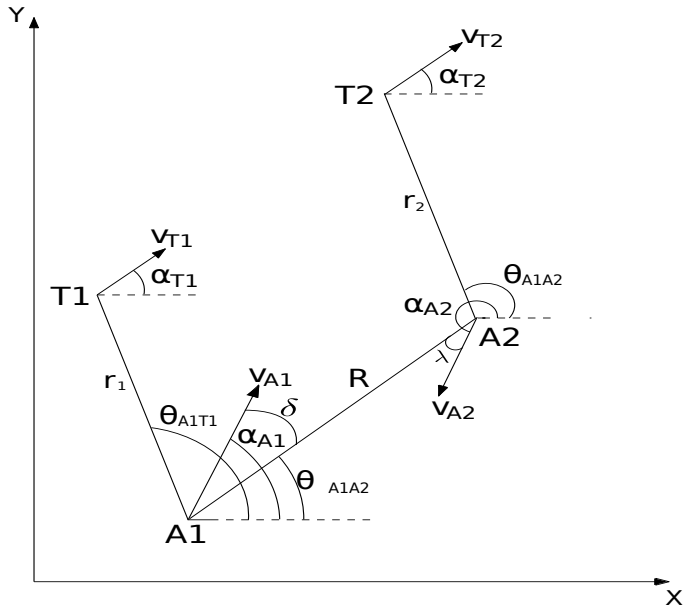}
	\caption{Four-agent engagement geometry. $A1$ and $A2$ are the attackers, $T1$ and $T2$ are the targets. }
	\label{fig:E_geometry}
\end{figure}

\subsection{Assumptions}
\begin{assumption}
We use a 2D Cartesian space, assuming that the altitude of the agents remain constant throughout the game.
\end{assumption}
\begin{assumption}\label{assump:constant_speed}
All the agents move with constant speed throughout the game.    
\end{assumption}
\begin{assumption}
Point mass models and only kinematic equations are considered for agents.
\end{assumption}
\begin{assumption}
Each target is pursued by one attacker, and each attacker pursues only one target.
\end{assumption}
\begin{assumption}
The target pair cooperate and act as a team, whereas it is assumed that the attackers do not exchange information and do not employ any inter-collision avoidance scheme.
\end{assumption}
\begin{assumption}
The target pair do not have information about the attacker states. They are estimated.
\end{assumption}
\begin{assumption}
The process and measurement noises $q(k) $ and $ \mu(k) $ are additive and zero-mean white Gaussian.
\end{assumption}
\begin{assumption}\label{assump:sensor}
	We assume that each agent (attacker or target) has onboard sensors to measure the range and LOS information of all other agents in the arena. 
\end{assumption}
\begin{assumption}
The target pair communicate with each other, sharing information.
\end{assumption}

\subsection{Engagement geometry}
The engagement geometry of the four agents is shown in Fig.~\ref{fig:E_geometry}. The equations of motion governing the relative motion of the four agents can be written as     
\begin{eqnarray}
\dot{x}_{T1}(t) &=&v_{T1}\cos{\alpha_{T1}}(t), \\
\dot{y}_{T1}(t) &=&v_{T1}\sin{\alpha_{T1}}(t), \\
\dot{x}_{T2}(t) &=&v_{T2}\cos{\alpha_{T2}}(t), \\
\dot{y}_{T2}(t) &=&v_{T2}\sin{\alpha_{T2}}(t), \\
\dot{x}_{A1}(t) &=&v_{A1}\cos{\alpha_{A1}}(t), \\
\dot{y}_{A1}(t) &=&v_{A1}\sin{\alpha_{A1}}(t), \\
\dot{x}_{A2}(t) &=&v_{A2}\cos{\alpha_{A2}}(t), \\
\dot{y}_{A2}(t) &=&v_{A2}\sin{\alpha_{A2}}(t), \\
\dot{R}(t) &=&v_{A2}\cos\left(\alpha_{A2}-\theta_{A1A2}\right) - v_{A1}\cos\left(\alpha_{A1}-\theta_{A1A2}\right),\\
\dot{r}_1(t) &=&v_{T1}\cos\left(\alpha_{T1}-\theta_{A1T1}\right) - v_{A1}\cos\left(\alpha_{A1}-\theta_{A1T1}\right),\\
\dot{r}_2(t) &=&v_{T2}\cos\left(\alpha_{T2}-\theta_{A2T2}\right) - v_{A2}\cos\left(\alpha_{A2}-\theta_{A2T2}\right),\\
\dot{\theta}_{A1A2}(t) &=&\frac{1}{R}\left( v_{A2}\sin\left(\alpha_{A2}-\theta_{A1A2}\right) - v_{A1}\sin\left(\alpha_{A1}-\theta_{A1A2}\right)\right),\\
\dot{\theta}_{A1T1}(t) &=&\frac{1}{r_1}\left( v_{T1}\sin\left(\alpha_{T1}-\theta_{A1T1}\right) - v_{A1}\sin\left(\alpha_{A1}-\theta_{A1T1}\right)\right),\\
\dot{\theta}_{A2T2}(t) &=&\frac{1}{r_2}\left( v_{T2}\sin\left(\alpha_{T2}-\theta_{A2T2}\right) - v_{A2}\sin\left(\alpha_{A2}-\theta_{A2T2}\right)\right),\\
\dot{\alpha}_{T1}(t) &=& u_1(t),\\
\dot{\alpha}_{T2}(t) &=& u_2(t),
\end{eqnarray}
where $ v_{T1},~v_{T2}$ are the velocities of the targets $T1$ and $T2$, $v_{A1},~v_{A2} $ are the velocities of the attackers $A1$ and $A2$, $\alpha_{T1}\left( t\right),~\alpha_{T2}\left( t\right)$ are the heading angles of the respective targets, and $\alpha_{A1}\left( t\right),~\alpha_{A2}\left( t\right)$ are the heading angles of the respective attackers. $\theta_{A1A2}(t)$ is the line-of-sight (LOS) angle between the attackers $A1$ and $A2$, and $ \theta_{A1T1}(t),~\theta_{A2T2}(t) $ are the LOS angles between the respective targets and the attackers. $R(t)$ is the distance between the attackers $ A1 $ and $ A2 $, $r_1(t)$ is the distance between $ A1 $ and $ T1 $ while $ r_2(t) $ is the distance between $ A2 $ and $ T2 $. $u_1(t)$ and $u_2(t)$ are the control inputs of $T1$ and $T2$. The angles $ \delta(t) $ and $ \lambda(t) $ are defined as
\begin{eqnarray}
\delta(t) &=&\alpha_{A1}-\theta_{A1A2} ,\\
\lambda(t) &=&\pi+\left( \alpha_{A2}-\theta_{A1A2}\right) .
\end{eqnarray}
All the states and heading angles of the agents are changing with respect to time $t$, and the notation $(t)$ is omitted in the rest of the
paper for readability. We will now formulate the NMPC framework using the defined assumptions and equations of motion.

\section{Nonlinear model predictive control formulation}\label{sec:nmpc}

\begin{figure}
	\centering
	\includegraphics[width=9cm]{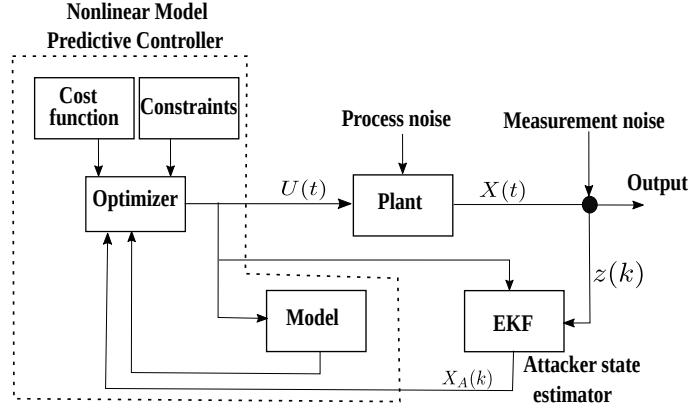}
	\caption{Nonlinear model predictive control scheme. }
	\label{fig:block_diagram}
\end{figure}

We propose a strategy wherein the target team uses a nonlinear model-predictive control (NMPC) scheme to compute their control commands so as to lure the attackers into a collision course. Model-predictive control is a state-of-the-art optimal control technique that considers future outcomes while calculating the current control input~\cite{mpcbook}. A mathematical model of the system under consideration is used to predict the future states and control actions up to a finite time known as the prediction horizon. After applying the first control input from the calculated sequence, the prediction window is shifted forward in time. This receding horizon approach combined with state-feedback helps in re-planning the control sequence to counter the uncertainties involved in real-world systems. Fig.~\ref{fig:block_diagram} shows the structure of the proposed NMPC scheme. The main components of the controller are a mathematical model of the system, a nonlinear convex optimizer, a cost function, and physical constraints on the states and controls. The plant represents the four-agent system, and the EKF is used for estimating the attacker states.  

The objective function for the NMPC can be formulated as:
\begin{equation}
\min_{\boldsymbol{\dot{\alpha}_{T1},\dot{\alpha}_{T2}} \in \mathcal{PC}(t,t+\tau_h)}J= \int_{t}^{t+\tau_{h}}\left( w_1 R+w_2\delta+w_3\lambda-w_4 r_1- w_5 r_2\right) \mathop{}\!\mathrm{d}t, 
\end{equation}
subject to:
\begin{eqnarray}
\dot{X}&=&f\left( X,U,t\right),\\
U&\in&\left[ U^{-},U^{+}\right],
\end{eqnarray}
where $ X=\{ x_{T1},y_{T1},x_{T2},y_{T2},x_{A1},y_{A1},x_{A2},y_{A2},R,r_1,r_2,\\ \theta_{A1A2},\theta_{A1T1},\theta_{A2T2},\alpha_{T1},\alpha_{T2}\} $ are the states, and $ U=\{ \dot{\alpha}_{T1},\dot{\alpha}_{T2}\} $ are the control commands computed for the targets $T1$ and $T2$. The $ U^{-} $ and $ U^{+} $ are the lower and upper bounds of $ U$, $ \mathcal{PC}(t,t+\tau_h) $ denotes the space of piece-wise continuous function defined over the time interval $ \left[ t,t+\tau_h\right]  $, and $ w_{i,i=1 \ldots 5} $ are the weights. The term $ R $ minimizes the distance between the attackers to make them collide, $ r_1 $ and $ r_2 $ are maximized so that the targets are not in a collision course with the attackers, and $ \delta $ and $ \lambda $ are minimized to keep the velocity vectors on the $ A1-A2 $ line-of-sight (LOS). Since the attacker states are unknown to the target pair, they are estimated using an EKF, which is formulated as follows.

\subsection*{Extended Kalman Filter}
The structure of EKF is given as~\cite{ekfbook}

\noindent{Model}\label{sec:ekf}
\begin{eqnarray}
	X_A(k)&=&f_A\left( X_A(k-1),U(k),k\right) +q(k),\\
	z(k)&=&h\left( X_A(k),U(k),k\right) +\mu(k),
\end{eqnarray}  
\noindent {Prediction}
\begin{align}
	X_A(k|k-1)&=f_A( X_A(k-1|k-1),U(k),k)\Delta t+X_A(k-1|k-1), \\
	P(k|k-1)&=\nabla F_{X_A}P(k-1|k-1)\nabla F_{X_A}'+Q,\\
	z(k|k-1)&=h( X_A(k|k-1)),
\end{align}
\noindent {Update}
\begin{eqnarray}
	X_A(k|k)&=&X_A(k|k-1)+K(k)\nu(k),  \\
	P(k|k)&=&P(k|k-1)-K(k)S(k)K'(k),\\
	\nu(k)&=&z(k)-z(k|k-1),\\
	K(k)&=&P(k|k-1)\nabla H_{X_A}'S^{-1}(k),\\
	S(k)&=&\nabla H_{X_A}P(k|k-1)\nabla H_{X_A}'+\Sigma,
\end{eqnarray}
where ${X_A}(k)$ and $z(k) $ represent the attacker state model and measurement model, respectively, and $P(k)$ is the estimation covariance matrix. $ q(k) $ and $ \mu(k) $ are the process and measurement noises, and $ Q $ and $ \Sigma $ are the state covariance and measurement covariance matrices. $ \nu(k) $ is the innovation parameter, and $ K(k) $ is the Kalman gain. The attacker states $\left\lbrace x_{A1},y_{A1},\alpha_{A1},x_{A2},y_{A2},\alpha_{A2}\right\rbrace $ and controls $ a_{A1},a_{A2} $ need to be estimated, hence the estimation model is represented with dynamics
\begin{equation}
f_A=\begin{bmatrix}
\dot{x}_{A1}\\
\dot{y}_{A1}\\
\dot{x}_{A2}\\
\dot{y}_{A2}\\
\dot{\alpha}_{A1}\\
\dot{\alpha}_{A2} \\
\dot{a}_{A1}\\
\dot{a}_{A2}
\end{bmatrix}=
\begin{bmatrix}
v_{A1}\cos\alpha_{A1}\\
v_{A1}\sin\alpha_{A1}\\
v_{A2}\cos\alpha_{A2}\\
v_{A2}\sin\alpha_{A2}\\
\frac{a_{A1}}{v_{A1}}\\
\frac{a_{A2}}{v_{A2}} \\
-a_{A1}\\
-a_{A2}
\end{bmatrix},
\end{equation}	
and the Jacobian of $ f_A $ is
\begin{equation}
\nabla F_{X_A}=\\
\begin{bmatrix}
0 & 0 & 0 & 0 & -v_{A1}\sin\alpha_{A1} & 0 & 0 & 0\\
0 & 0 & 0 & 0 &  v_{A1}\cos\alpha_{A1} & 0 & 0 & 0\\
0 & 0 & 0 & 0 & 0 & -v_{A2}\sin\alpha_{A2} & 0 & 0\\
0 & 0 & 0 & 0 & 0 &  v_{A2}\cos\alpha_{A2} & 0 & 0\\
0 & 0 & 0 & 0 & 0 & 0 &\frac{1}{v_{A1}} & 0\\
0 & 0 & 0 & 0 & 0 & 0 & 0 & \frac{1}{v_{A2}}\\
0 & 0 & 0 & 0 & 0 & 0 & -1 & 0 \\
0 & 0 & 0 & 0 & 0 & 0 & 0 & -1
\end{bmatrix}.
\end{equation} 
According to Assumption~\ref{assump:sensor}, the quantities $ R,r_1,r_2,\theta_{A1A2},\theta_{A1T1} $ and $ \theta_{A2T2} $ can be measured and the measurement model is given as
{\small
\begin{equation}
h = \begin{bmatrix}
\sqrt{\left( x_{T1}-x_{A1}\right) ^{2}+\left( y_{T1}-y_{A1}\right) ^{2}}\\
\sqrt{\left( x_{T2}-x_{A2}\right) ^{2}+\left( y_{T2}-y_{A2}\right) ^{2}}\\
\sqrt{\left( x_{A2}-x_{A1}\right) ^{2}+\left( y_{A2}-y_{A1}\right) ^{2}}\\
\tan^{-1} \left( \frac{y_{T1}-y_{A1}}{x_{T1}-x_{A1}}\right)\\
\tan^{-1} \left( \frac{y_{T2}-y_{A2}}{x_{T2}-x_{A2}}\right)\\
\tan^{-1} \left( \frac{y_{A2}-y_{A1}}{x_{A2}-x_{A1}}\right)
\end{bmatrix}.
\end{equation}}
The Jacobian of the measurement model is given by \eqref{eq:jh}.
\begin{table*}[h!]
\begin{equation}
\nabla H_{X_A}=\\ \begin{bmatrix}
\frac{-(x_{T1}-x_{A1})}{r_1}&
\frac{-(y_{T1}-y_{A1})}{r_1}& 0 & 0 & 0 & 0 & 0 & 0\\
0 & 0 & \frac{-(x_{T2}-x_{A2})}{r_2}&
\frac{-(y_{T2}-y_{A2})}{r_2}& 0 & 0 & 0 & 0\\
\frac{-(x_{A2}-x_{A1})}{R} & \frac{-(y_{A2}-y_{A1})}{R} & \frac{x_{A2}-x_{A1}}{R} & \frac{y_{A2}-y_{A1}}{R} & 0 & 0 & 0 & 0\\
\frac{y_{T1}-y_{A1}}{r_1^2}&
\frac{-(x_{T1}-x_{A1})}{r_1^2}& 0 & 0 & 0 & 0 & 0 & 0\\
0 & 0 & \frac{y_{T2}-y_{A2}}{r_2^2}&
\frac{-(x_{T2}-x_{A2})}{r_2^2}& 0 & 0 & 0 & 0\\
\frac{y_{A2}-y_{A1}}{R^2} & \frac{-(x_{A2}-x_{A1})}{R^2} & \frac{-(y_{A2}-y_{A1})}{R^2} & \frac{x_{A2}-x_{A1}}{R^2} & 0 & 0 & 0 & 0
\end{bmatrix}.\label{eq:jh}
\end{equation}
\end{table*}
The NMPC scheme for the 2T2A problem formulated so far does not give us the answer to the question of whether the targets would be captured or not given the initial positions of the attackers and the targets. In reality, the formulation would give optimal trajectories for the target survival only in a subset of the complete game, where the targets are guaranteed to be successful. In the next section, we derive the conditions that would help us determine the answer to the escape problem. 

\section{Escape region}\label{sec:escape}
Here we formulate the 2T2A problem as a game of kind~\cite{isaacsbook}, where the outcome of target capture or escape could be determined by the initial position of the agents. We use the concept of Apollonius circles to determine the escape region for the targets in the Cartesian plane subject to the following assumptions.
\begin{assumption}\label{assump:attacker_speed}
	The attackers are identical and have equal speed.
\end{assumption} 
\begin{assumption}
	The targets have equal speed and are slower than the attackers. Otherwise, the targets can always evade the attackers.
\end{assumption}
\begin{assumption}\label{assump:straight_line}
    For ease of analysis we assume that all the agents travel in straight  paths.
\end{assumption}
\begin{assumption}
	The attackers are assumed to pursue the targets closer to them. i.e., $ x_{T1}>0 $ and $ x_{T2}<0 $. 
\end{assumption}

To make the analysis easier, we modify the engagement geometry reference frame, as shown in Fig.~\ref{fig:frame}. The $x-$axis is defined as the line joining the coordinates of the attackers, $ A1 $ and$ A2 $, and the $y-$axis is defined as the perpendicular bisector of the line segment $ \overline{A1A2} $, where $ A1,A2,T1,T2 $ represent the positions of attacker-1, attacker-2, target-1, and target-2 respectively.  
\begin{figure}
	\centering
	\includegraphics[width=8cm,height =8cm]{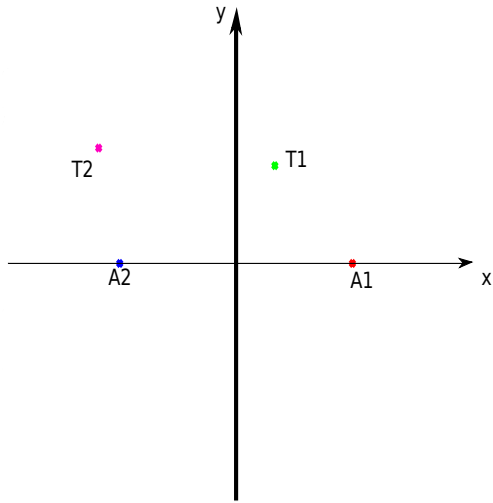}
	\caption{Modified reference frame. }
	\label{fig:frame}
\end{figure} 


\begin{lemma}\label{lemma:y_axis}
	The locus of points where the attackers $ A1$ and $A2 $ can reach simultaneously is represented by the $y-$axis in the modified reference frame.
\end{lemma}
\begin{proof}
Based on Assumptions~\ref{assump:constant_speed}, \ref{assump:attacker_speed}, and \ref{assump:straight_line}, it is  well-known  that the locus of points where the two agents can reach simultaneously is the orthogonal bisector of the line segment joining the agents~\cite{isaacsbook}. Since $y-$axis is defined as the orthogonal bisector of the line segment $ \overline{A1A2} $, it is the locus of points where the attackers $ A1$ and $A2 $ can reach simultaneously.
\end{proof}

For finding the escape region for the targets, the 2T2A game is divided into two sub-games containing three agents each. These new sub-games will contain (i) attackers A1--A2 and target T1 and (ii) attackers A1--A2 and target T2.

\begin{lemma}\label{lemma:a1_t1_condition}
	The attacker $A2$ will collide with the attacker $A1$  if the $ A1-T1 $ Apollonius circle intercepts the $ y- $axis.
\end{lemma}
\begin{proof}
The Apollonius circle for the $A1-T1$ engagement is constructed using the center and radius given by
\begin{equation}
	\left( x_c,y_c\right) _{A1T1} = \left(  \frac{x_{T1}-\gamma_{A1T1}^2x_{A1}}{1-\gamma_{A1T1}^2},\frac{y_{T1}}{1-\gamma_{A1T1}^2}\right), \label{c_A1T1}
\end{equation}
and 
\begin{equation}
	r_{A1T1} = \frac{\gamma_{A1T1}\sqrt{ \left( x_{T1}-x_{A1}\right)^2 +y_{T1}^2}}{1-\gamma_{A1T1}^2}, \label{r_A1T1}
\end{equation}
	where $ \gamma_{A1T1} $ is the speed ratio defined by 
	\begin{equation}
		\gamma_{A1T1} = \frac{v_{T1}}{v_{A1}}.
	\end{equation}
The $A1-T1$ circle represents the points at which $A1$ and $T1$ can reach simultaneously. Since the $ y- $axis represents the points at which $A1$ and $A2$ can reach simultaneously according to Lemma~\ref{lemma:y_axis}, the attacker $A2$ will collide with the attacker $A1$ only if the $ A1-T1 $ Apollonius circle intercepts the $ y- $axis.
\end{proof}


\begin{figure}
	\centering 
	\begin{subfigure}{0.5\textwidth}
		\includegraphics[width=\linewidth]{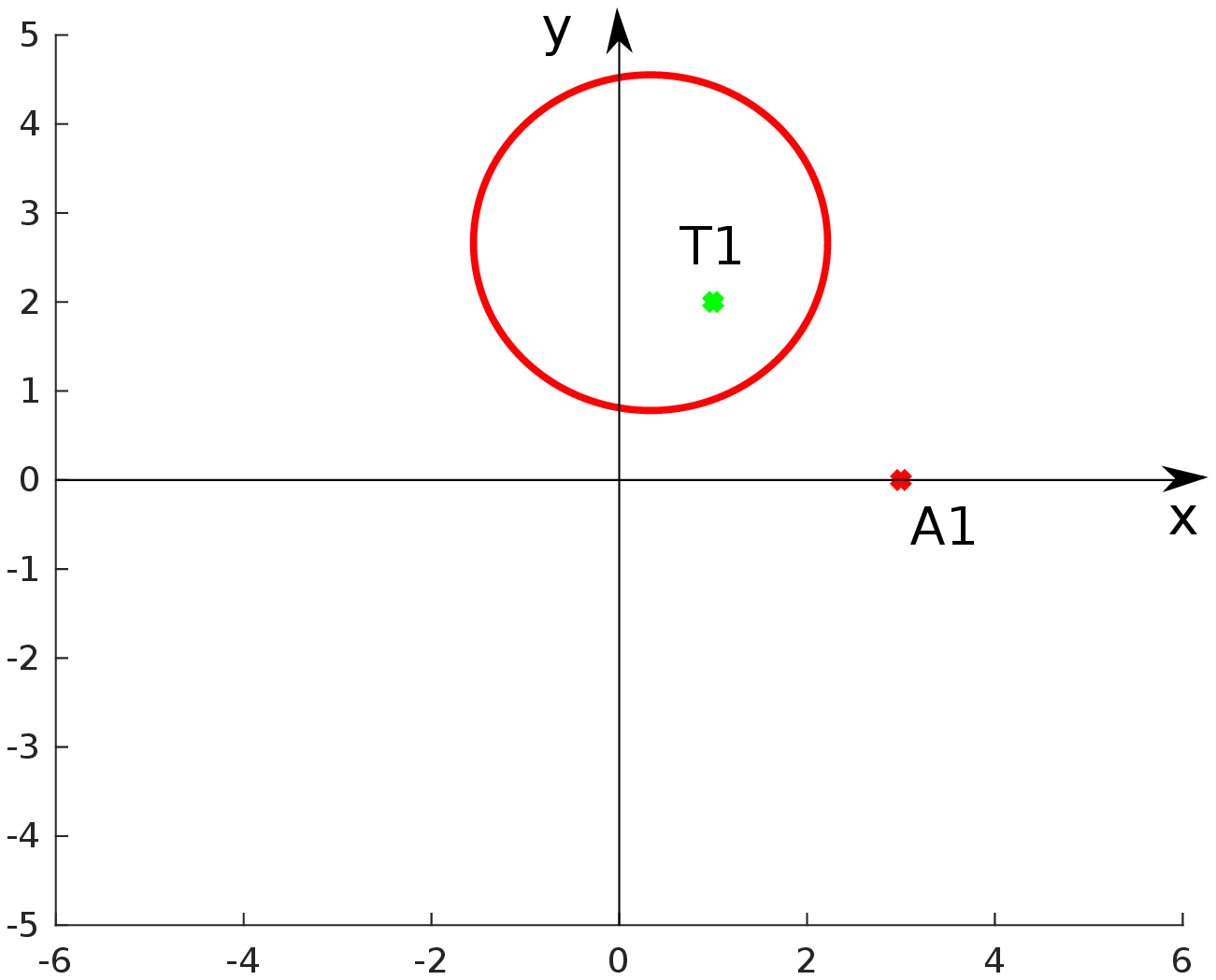}
		\caption{}
		\label{fig:A1T1}
	\end{subfigure}\hfil 
	\begin{subfigure}{0.5\textwidth}
		\includegraphics[width=\linewidth]{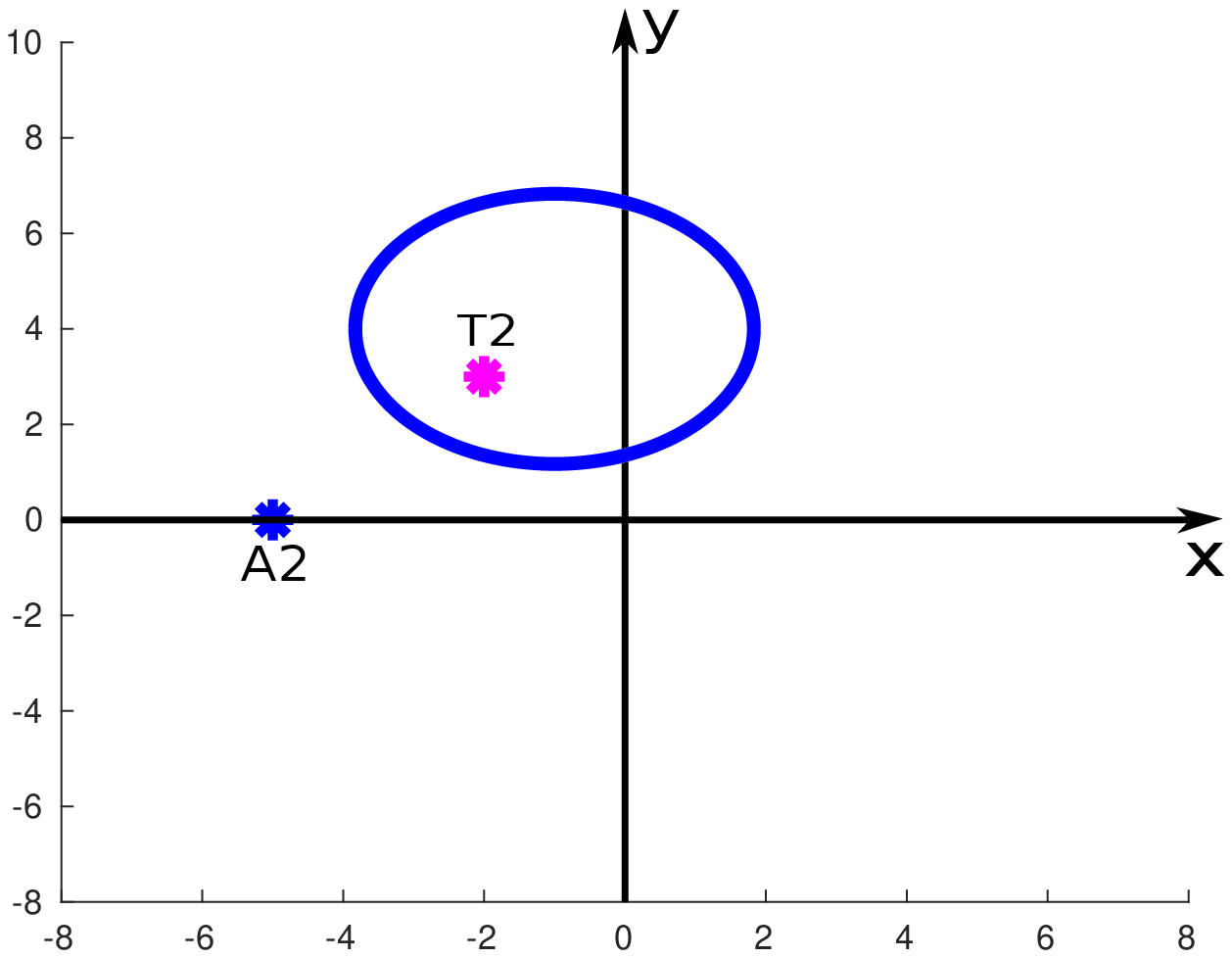}		
		\caption{}
		\label{fig:A2T2}
	\end{subfigure}\hfil 
	\caption{Apollonius circles for the attacker-target engagements. (a) $ A1-T1 $ engagement. (b) $ A2-T2 $ engagement.}
	\label{fig:Apollonis_circles}
\end{figure} 

\begin{lemma}\label{lemma:a2_t2_condition}
	The attacker $A1$ will collide with the attacker $A2$ only if the $ A2-T2 $ Apollonius circle intercepts the $ y- $axis.
\end{lemma}
\begin{proof}
	The Apollonius circle for the $A2-T2$ engagement is constructed using the center and radius given by
	\begin{equation}
		\left( x_c,y_c\right) _{A2T2} = \left(  \frac{x_{T2}-\gamma_{A2T2}^2x_{A2}}{1-\gamma_{A2T2}^2},\frac{y_{T2}}{1-\gamma_{A2T2}^2}\right), \label{c_A2T2}
	\end{equation}
	and
	\begin{equation}
		r_{A2T2} = \frac{\gamma_{A2T2}\sqrt{ \left( x_{T2}-x_{A2}\right)^2 +y_{T2}^2}}{1-\gamma_{A2T2}^2}, \label{r_A2T2}
	\end{equation}
		where $ \gamma_{A2T2} $ is the speed ratio defined by 
	\begin{equation}
		\gamma_{A2T2} = \frac{v_{T2}}{v_{A2}}.
	\end{equation}
The $A2-T2$ circle represents the points at which $A2$ and $T2$ can reach simultaneously. Similarly, according to Lemma~\ref{lemma:y_axis}, the $y-$axis represents the points at which $A1$ and $A2$ can reach simultaneously. Hence, the attacker $A1$ will collide with the attacker $A2$ only if the $ A2-T2 $ circle intercepts the $ y- $axis.
\end{proof} 

The Fig.~\ref{fig:A1T1} represents an example Apollonius circle for the $ A1-T1 $ engagement, and Fig.~\ref{fig:A2T2} for the $ A2-T2 $ engagement. Now, let us state the theorem for finding the target escape region for the 2T2A game.

\begin{theorem}\label{thrm:escape}
The escape region for the targets exist if and only if the two Apollonius circles defined by $ A1-T1 $ and $ A2-T2 $ engagement have a common interception point with the $y-$axis. i.e., the target pair will be able to successfully lure the attackers into a collision if and only if the following conditions are satisfied:
\begin{enumerate}
	\item $Y_1 \cap Y \neq \emptyset $,
	\item $ Y_2 \cap Y \neq \emptyset $,
	\item $ Y_1 \cap Y_2 \neq \emptyset $,
\end{enumerate} 
where $ Y $ is the set of all points on the $y-$axis, $ Y_1 $ is the set of all points on the $ A1-T1 $ Apollonius circle, $ Y_2 $ is the set of all points on the $ A2-T2 $ Apollonius circle, and $ \emptyset $ is the null set.	
\end{theorem}
\begin{proof}
    We prove the theorem using  Lemma~\ref{lemma:y_axis}, \ref{lemma:a1_t1_condition}, and \ref{lemma:a2_t2_condition}.  According to Lemma~\ref{lemma:y_axis} and \ref{lemma:a1_t1_condition}, the $A1-T1$ Apollonius circle should intercept the $y-$axis for $A2$ to intercept $A1$. If $ Y_1 \cap Y = \emptyset $, the $ A1-T1 $ Apollonius circle will not intersect the $y-$axis, and hence the target $T1$ will be captured before it can cross the $y-$axis and escape. Therefore $1) Y_1 \cap Y \neq \emptyset $ must be satisfied.
	
	From Lemma~\ref{lemma:y_axis} and \ref{lemma:a2_t2_condition}, for $T2$ to escape, the $A2-T2$ circle should intercept the $y-$axis so that $A1$ can capture $A2$. This condition will not be satisfied if $ Y_2 \cap Y = \emptyset $. If any one of the targets is captured, eventually, the other one will also be captured. Therefore the condition $ 2) Y_2 \cap Y \neq \emptyset $ must be satisfied.
	
	Since each target depends on the other attacker for survival, the $A1-T1$ and $A2-T2$ circles should intercept each other on the $y-$ axis for $A1-A2$ collision. Otherwise, if $ Y_1 \cap Y_2 = \emptyset $, even if conditions 1 and 2 are satisfied, $ A1 $ and $ A2 $ will cross the $ y-$axis at different points, and hence they will not collide. Therefore, for the attackers to collide and the targets to escape, conditions 1, 2, and 3 should be simultaneously satisfied. 
\end{proof}

\begin{figure}
	\centering
	\includegraphics[width=10cm]{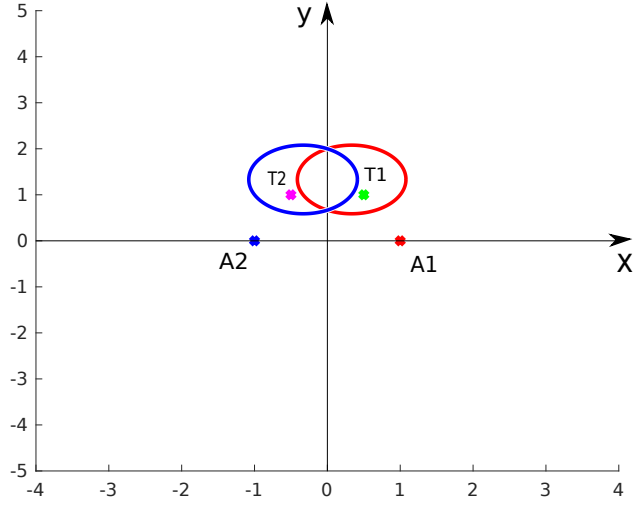}
	\caption{Apollonius circles for the 2T2A game.}
	\label{fig:A1A2T1T2}
\end{figure} 
From Fig.~\ref{fig:A1A2T1T2}, the equation of the $A1-T1$ Apollonius circle can be written as
\begin{equation}
	\left( x- \left( \frac{x_{T1}-\gamma_{A1T1}^2x_{A1}}{1-\gamma_{A1T1}^2} \right)\right)^2 + \left( y - \left( \frac{y_{T1}-\gamma_{A1T1}^2y_{A1}}{1-\gamma_{A1T1}^2} \right) \right)^2 = \left( \frac{\gamma_{A1T1}\sqrt{ \left( x_{T1}-x_{A1}\right)^2 +\left( y_{T1}-y_{A1}\right)^2}}{1-\gamma_{A1T1}^2} \right)^2.
\end{equation} 
Since $A1$ lies on the $x-$axis, $y_{A1}=0$. For finding the intersection points of $A1-T1$ circle with the $y-$axis, we put $x=0$, and the following expression is obtained.
\begin{eqnarray}
	&y^2-2y\left( \frac{y_{T1}}{1-\gamma_{A1T1}^2}\right) +\left( \frac{x_{T1}-\gamma_{A1T1}^2x_{A1}}{1-\gamma_{A1T1}^2}\right) ^2 +\left( \frac{y_{T1}}{1-\gamma_{A1T1}^2}\right) ^2 \nonumber\\
	&-\frac{\gamma_{A1T1}^2\left( \left( x_{T1}-x_{A1}\right) ^2+y_{T1}^2\right) }{\left( 1-\gamma_{A1T1}^2\right) ^2} = 0. 
	\label{eqn:intercept}
\end{eqnarray}
Interception points of $ A1-T1 $ Apollonius circle with the $y-$axis are the solutions of this quadratic equation~(\ref{eqn:intercept}) and are represented by $ \overline{y}_1 $ and $ \underline{y}_1 $. Similarly, the interception points of $ A2-T2 $ Apollonius circle with the $y-$axis can be represented using $ \overline{y}_2 $ and $ \underline{y}_2 $. The escape region for the targets can be mapped by verifying the conditions given in Theorem~\ref{thrm:escape} by analyzing the positions of $y-$intercepts: $ \overline{y}_1,  \underline{y}_1, \overline{y}_2,$ and $ \underline{y}_2 $. 

The escape region of the targets for an example initial configuration is given in Fig.~\ref{fig:escape_region_T1T2}. The initial positions of the agents were selected as $ x_{A1}=100,~y_{A1}=0,~x_{A2}=-100,~y_{A2}=0,~x_{T1}=150,~y_{T1}=600 $, and $\gamma_{A1T1}=\gamma_{A2T2}=~0.5$. The initial position of $T2$ is varied between $-1000<x_{T2}<0$ and $0<y_{T2}<2500 $. At each selected initial point of $T2$, the conditions stated in Theorem~\ref{thrm:escape} is checked by comparing the $y-$intercepts: $ \overline{y}_1,  \underline{y}_1, \overline{y}_2,$ and $ \underline{y}_2 $. The initial points of $T2$ which satisfy the conditions are marked with black color. Finally, after completion of this map, we will be able to tell that if the target $T2$'s initial position lies in the escape region marked by black color, then the targets would escape, and if it lies outside of the escape region, the targets would be captured.  

\begin{figure}
	\centering
	\includegraphics[width = 10cm]{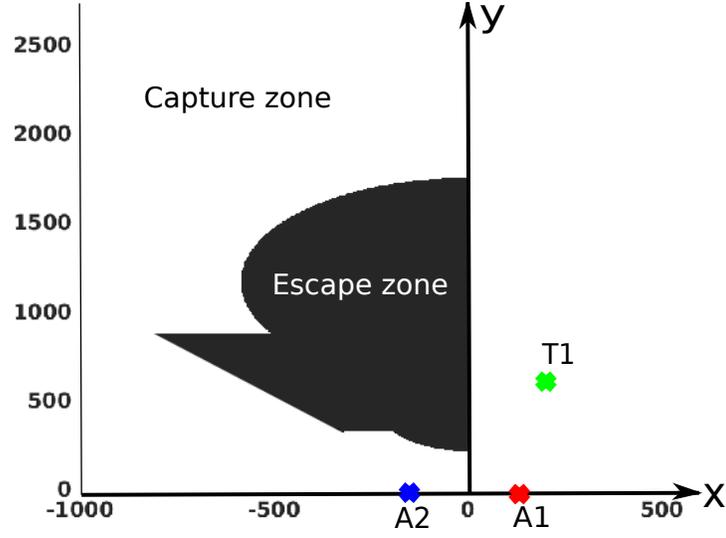}
	\caption{Escape region for the targets. }
	\label{fig:escape_region_T1T2}
\end{figure}

\section{Results and Discussion}\label{sec:results}
The performance of the proposed NMPC formulation was evaluated through numerical simulations. Initially, we will describe the simulation setup followed by a general example. Then two examples using the theoretical analysis given in Section~\ref{sec:escape} are presented to show the attackers colliding and missing.

\subsection{Simulation setting}
The simulations were carried out using CasADi-Python~\cite{Andersson2019}. The horizon, $ \tau_{h} $ is $2.5$\,s for all simulations (prediction window of 50 steps with a sampling time of $0.05$\,s). The state covariance matrix $ Q $ and the measurement covariance matrix $ \Sigma $ for the EKF are selected as
\begin{equation}
Q=\begin{bmatrix}
0.1 & 0 & 0 & 0 & 0 &0 & 0 & 0 \\
0 & 0.1 & 0 & 0 & 0 &0 & 0 & 0\\
0 & 0 & 0.1 & 0 & 0 &0 & 0 & 0\\
0 & 0 & 0 & 0.1 & 0 &0 & 0 & 0\\
0 & 0 & 0 & 0 & 0.01 &0 & 0 & 0\\
0 & 0 & 0 & 0 & 0 &0.01 & 0 & 0\\
0 & 0 & 0 & 0 & 0 & 0 & 0.1 & 0\\
0 & 0 & 0 & 0 & 0 & 0 & 0 & 0.1\\  
\end{bmatrix},
\end{equation}
\begin{equation}
\Sigma=\begin{bmatrix}
0.1 & 0 & 0 & 0 & 0 & 0 \\
0 & 0.1 & 0 & 0 & 0 & 0 \\
0 & 0 & 0.1 & 0 & 0 & 0 \\
0 & 0 & 0 & 0.01 & 0 & 0 \\
0 & 0 & 0 & 0 & 0.01 & 0 \\
0 & 0 & 0 & 0 & 0 & 0.01 
\end{bmatrix}.
\end{equation} 
The capture radii are taken as $ R_{c}=1$\,m, $ r_{1_{c}}=1 $\,m, and $r_{2_{c}}=1$\,m. The weights were selected as $w_1 = 30,w_2=w_3=60,w_4=w_5=1$.  The angular velocities are constrained to $-0.785\leq \{\dot{\alpha}_{T1},\dot{\alpha}_{T2}\} \leq 0.785$\,rad/s due to the practical considerations on the turn rate of the agents. The velocities of the agents are taken as $ v_{A1},v_{A2}=60 $\,m/s and $ v_{T1},v_{T2}=30 $\, m/s. 


\begin{figure}
	\centering 
	\begin{subfigure}{0.5\textwidth}
		\includegraphics[width=\linewidth]{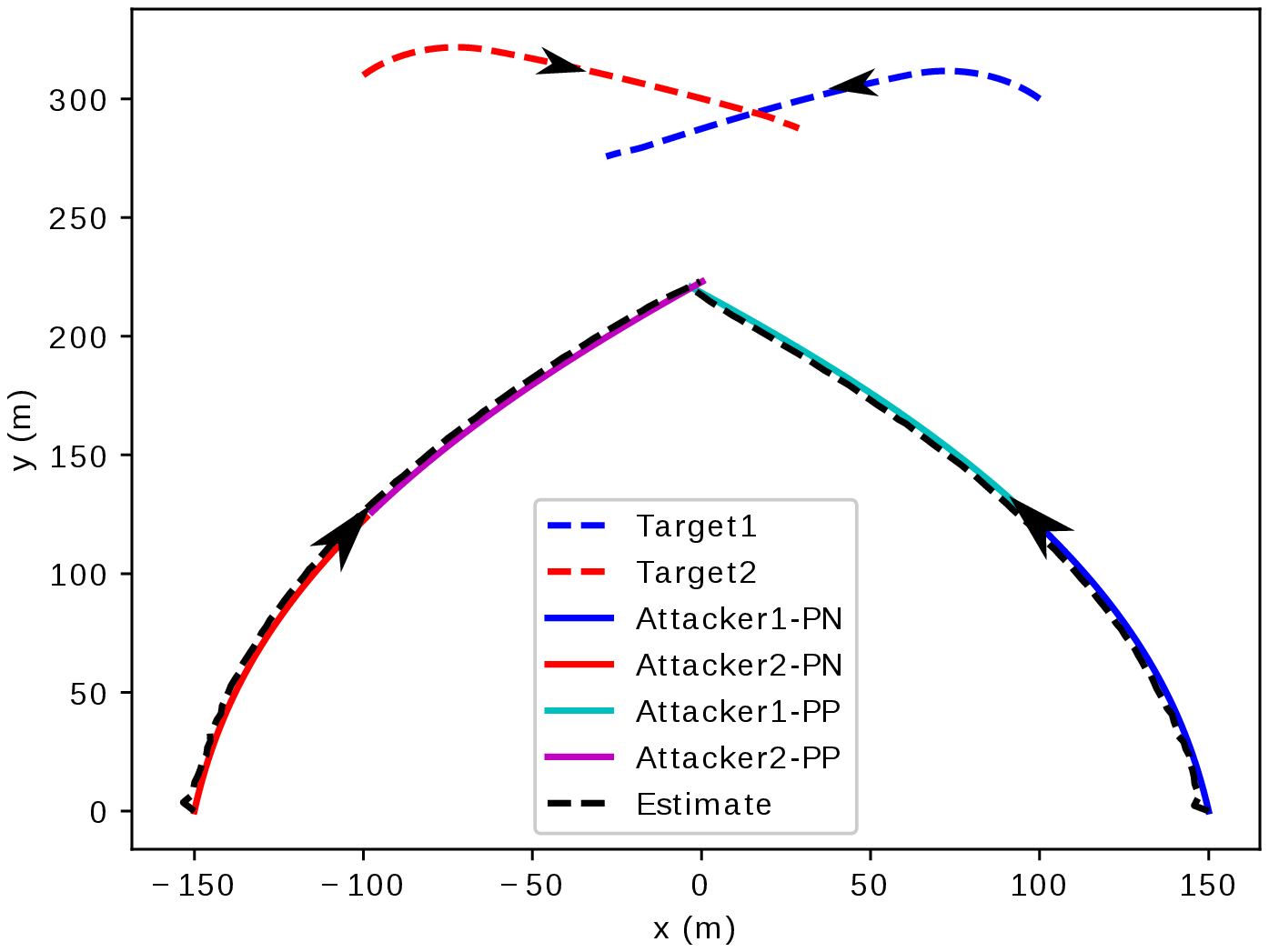}
		\caption{}
		\label{fig:j3_traj}
	\end{subfigure}\hfil 
	\begin{subfigure}{0.45\textwidth}
		\includegraphics[width=\linewidth]{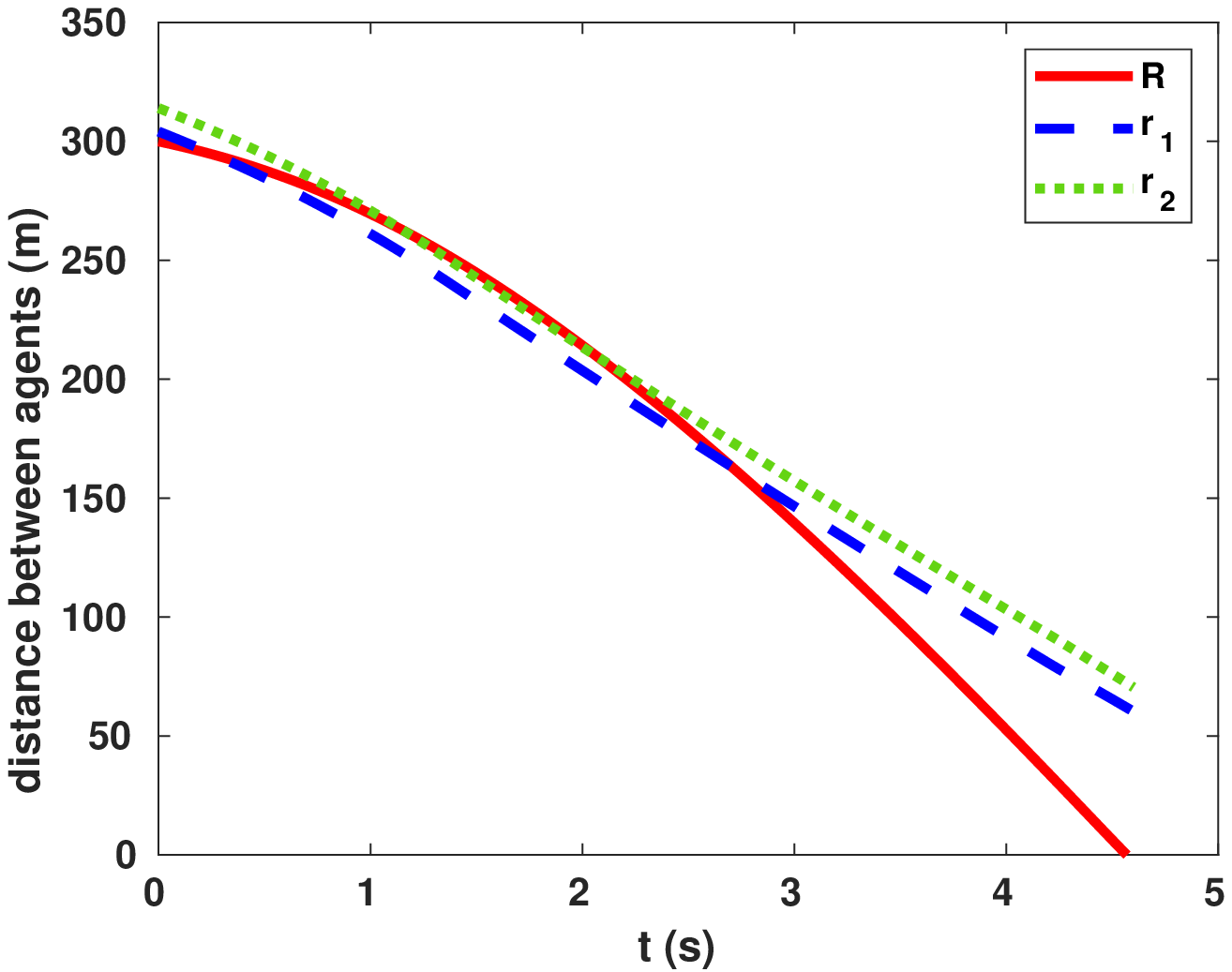}		
		\caption{}
		\label{fig:j3_R}
	\end{subfigure}\hfil 
	\caption{An example scenario of the four-agent game. (a) The trajectories of the two targets and the two attackers. (b) Distance between different agents.}
	\label{fig:example_sim}
\end{figure} 

\subsection{A general example scenario}
We consider an engagement scenario as shown in Fig.~\ref{fig:j3_traj}. The figure shows the agent trajectories for an initial configuration of $ x_{T1}=100,y_{T1}=300,x_{T2}=-100,y_{T2}=310,x_{A1}=150,y_{A1}=0,x_{A2}=-150,y_{A2}=0 $. In the simulation, the attackers switch between proportional navigation (PN) guidance law and pure pursuit (PP) guidance. This is done in order to show the efficacy of the NMPC scheme against unknown attacker guidance laws. The attackers switch their guidance law between PN and PP every 2.5s. The PP guidance law is given as  $	a_{A1}=-\kappa\left( \alpha_{A1}-\theta_{A1T1}\right), a_{A2}=-\kappa\left( \alpha_{A2}-\theta_{A2T2}\right) $ \cite{siouris2004missile} 
and the PN guidance is given as    
$	a_{A1}=Nv_{A1}\dot{\theta}_{A1T1},a_{A2}=Nv_{A2}\dot{\theta}_{A2T2}$ \cite{c13}. 
The navigation constant, $N$ = 3, is taken for the PN guidance law and $\kappa$ = 2 for the PP law. The target pair was able to lure the attackers into a collision successfully. The evolution of distances between the agents is shown in Fig.~\ref{fig:j3_R}. The distance between the attackers went to zero, confirming the collision. Control input profiles of the targets are given in Fig.~\ref{fig:j3_control}. It can be seen that the bounds on the inputs were strictly followed. The estimator performance is shown in Fig.~\ref{fig:j3_error}. The estimation errors in all the attacker states are very low and stay within the 3$ \sigma $ bounds. Hence in the following examples, we do not show the estimation errors. 

\begin{figure}
	\centering
	\includegraphics[width=8cm]{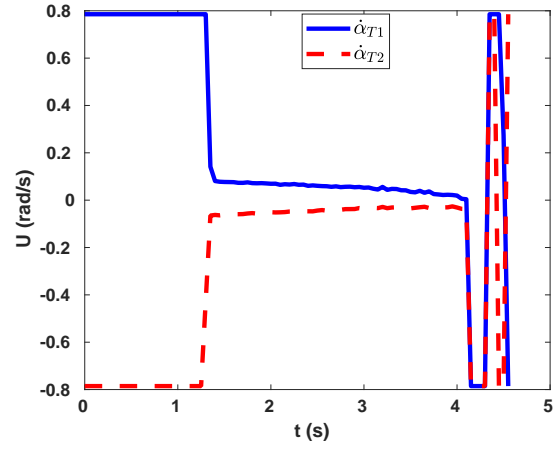}
	\caption{Control profile of the target pair. }
	\label{fig:j3_control}
\end{figure}     

\begin{figure*}
	\centering
	\includegraphics[width = 16cm]{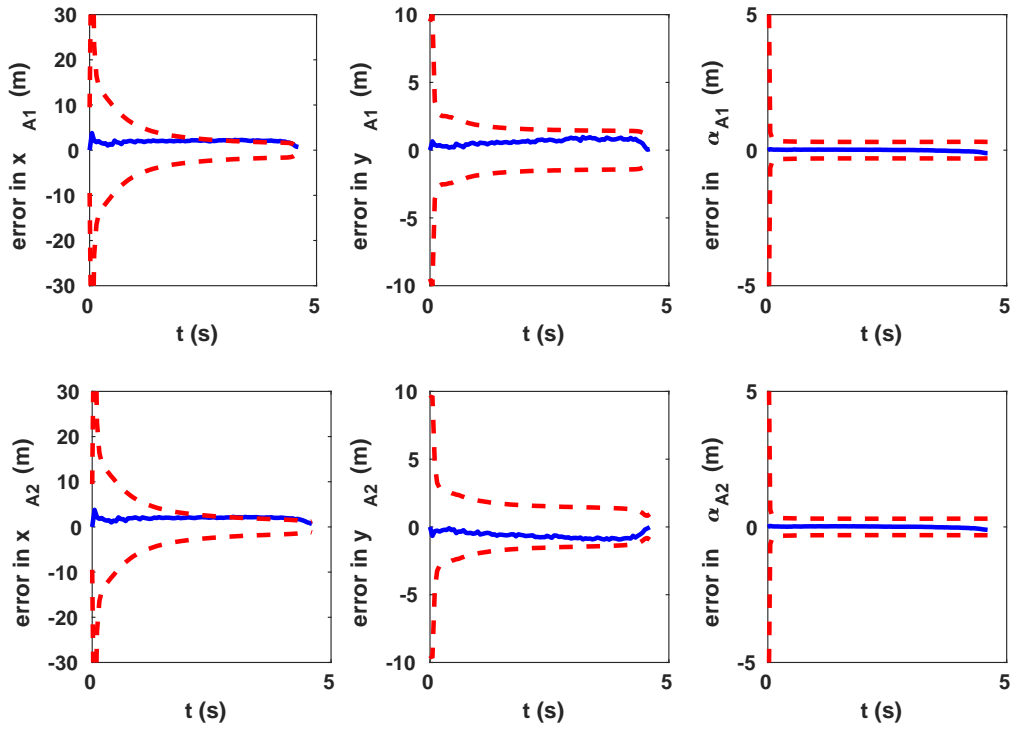}
	\caption{Estimation errors of the attacker states. }
	\label{fig:j3_error}
\end{figure*}

\subsection{Examples to validate the theoretical analysis}
We now show two examples based on the escape region determined in Fig.~\ref{fig:escape_region_T1T2}. We fix all the positions for the attackers and the target $T1$, while the target T2 position is changed in the simulations. 
First, the initial position of $ T2 $ is selected  inside the escape region as $ x_{T2}=-100,y_{T2}=1000 $ and is represented by $T2_e$ in Fig.~\ref{fig:escape_map_T2}. The agent trajectories for this example are shown in Fig.~\ref{fig:sim6_traj}, where we can see the attackers colliding. Fig.~\ref{fig:sim6_R} shows the distance ($R$) between the attackers converging to zero. The attackers collide with each other, and the targets escaped as expected. 

Next, the initial position of the target $ T2 $ is considered outside the escape region as $ x_{T2}=-1000,y_{T2}=2000 $ and is represented by $T2_c$ in Fig.~\ref{fig:escape_map_T2}. Under this initial condition, Fig.~\ref{fig:sim5_traj} shows that one of the targets was captured, and as the engagement continues, the other target will also be captured. Fig.~\ref{fig:sim5_R} shows the distance between the agents. $r_1$ goes to zero, confirming the capture of $T1$ by $A1$. 

\begin{figure}
	\centering
	\includegraphics[width=9cm]{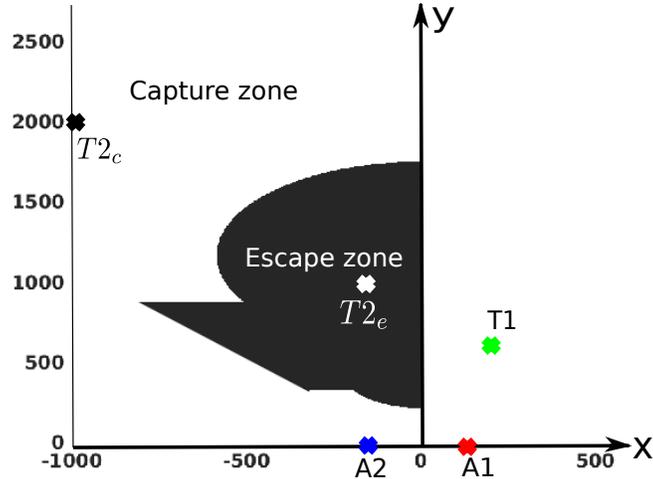}
	\caption{Initial agent configurations in the escape map. }
	\label{fig:escape_map_T2}
\end{figure}    


\begin{figure}
	\centering 
	\begin{subfigure}{0.5\textwidth}
		\includegraphics[width=\linewidth]{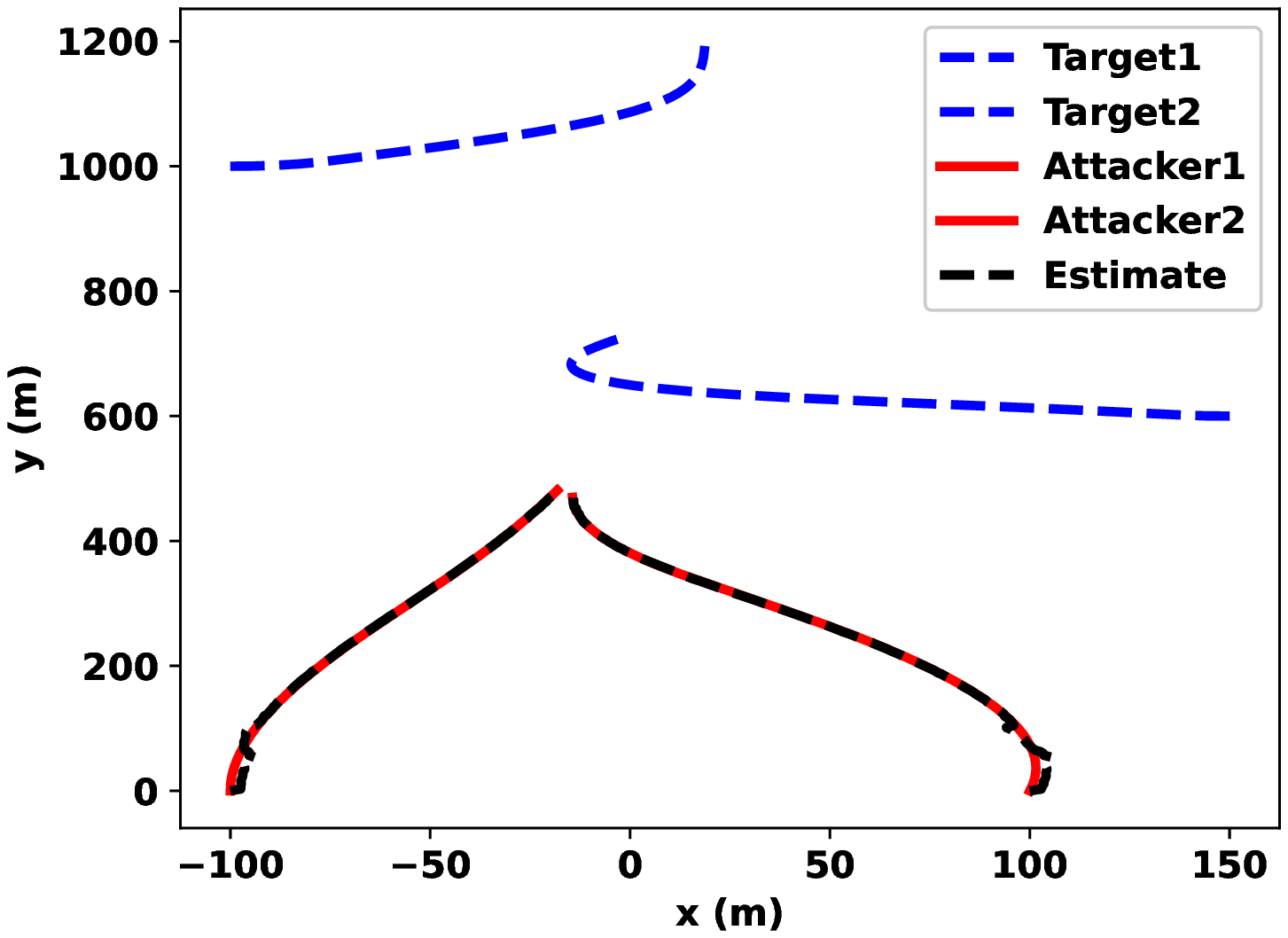}
		\caption{}
		\label{fig:sim6_traj}
	\end{subfigure}\hfil 
	\begin{subfigure}{0.5\textwidth}
		\includegraphics[width=\linewidth]{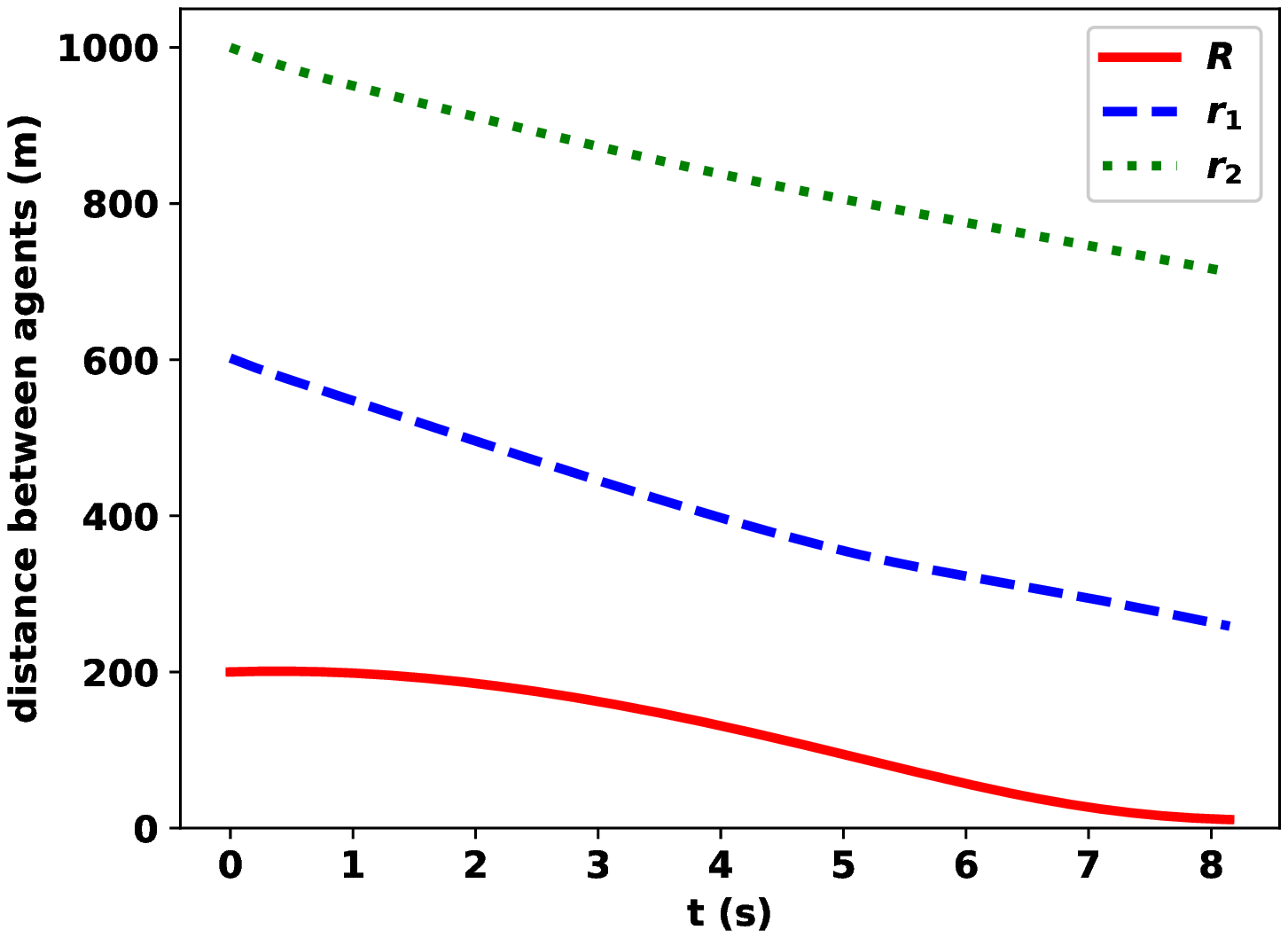}		
		\caption{}
		\label{fig:sim6_R}
	\end{subfigure}\hfil 
	\caption{Example escape scenario. (a) Agent trajectories. (b) Distance between the agents.}
	\label{fig:escape_sim}
\end{figure} 


\begin{figure}
	\centering 
	\begin{subfigure}{0.5\textwidth}
		\includegraphics[width=\linewidth]{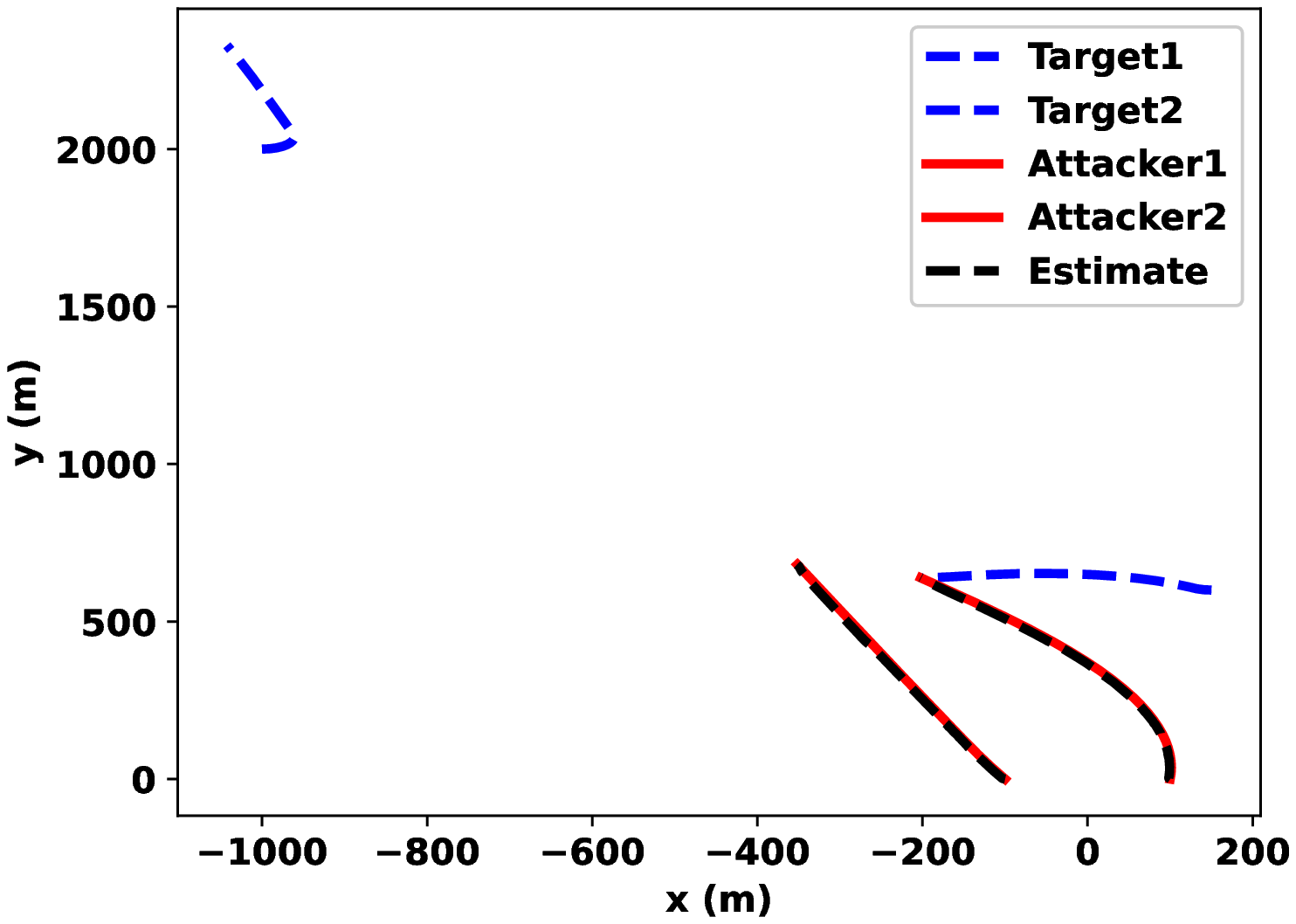}
		\caption{}
		\label{fig:sim5_traj}
	\end{subfigure}\hfil 
	\begin{subfigure}{0.5\textwidth}
		\includegraphics[width=\linewidth]{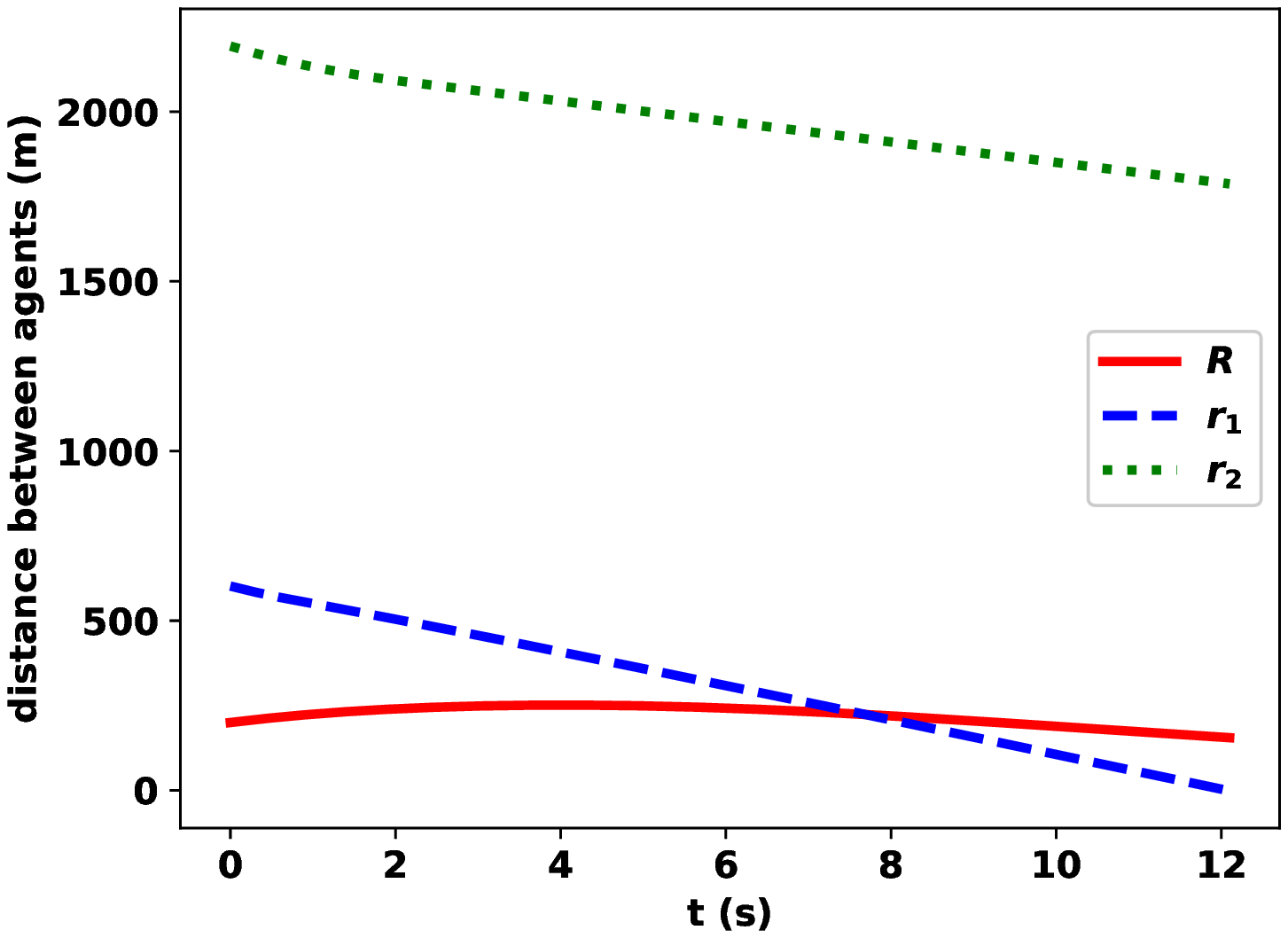}		
		\caption{}
		\label{fig:sim5_R}
	\end{subfigure}\hfil 
	\caption{Example capture scenario. (a) Agent trajectories. (b) Distance between the agents.}
	\label{fig:capture_sim}
\end{figure} 
	\section{Conclusions}\label{sec:conclusions}
In this paper, we proposed a cooperative strategy based on NMPC for the active defense of the targets in a two--targets two--attackers (2T2A) game. The NMPC computes control commands for the cooperative target pair against individually acting attackers, whose state information was estimated using an EKF. The problem was formulated as a game of kind that enables one to determine whether the targets would escape or not given the initial conditions of the attackers and the targets. The theoretical analysis was performed using Apollonius circles. The efficacy of the proposed scheme was validated using numerical simulations. The results show the ability of the NMPC to determine control commands satisfying the state and control constraints. Also, the integration of the EKF with the NMPC performed well as the estimation errors are within the $3\sigma$ bounds. The results also support the theoretical analysis. 

A natural extension of the proposed framework is to formulate the problem in 3D for aerial vehicles. Another extension is to develop a general framework for multiple attackers and multiple targets. Further, the work can be extended to multiple target--attacker--defender games. Also, optimal agent allocation strategies can be devised to maximize the number of targets escaped.

	
	\typeout{}
	\bibliographystyle{IEEEtran}
	\bibliography{ref}
	
\end{document}